\RequirePackage{fix-cm}
\documentclass{svjour3}                     

\smartqed
\usepackage{graphicx}
\usepackage{enumitem}
\usepackage{color}

\usepackage{amsfonts, amssymb, amsmath, amscd, latexsym}
\usepackage{xcolor}
\usepackage[all]{xy}
\usepackage{enumitem}
\newcommand{\N}{\mathbb{N}}
\newcommand{\R}{\mathbb{R}}
\setlist[enumerate]{wide=\parindent}

\journalname{Nonlinear Dynamics}

\begin{document}

\title{Chaos control in the fractional order logistic map via impulses}
\vspace{5mm}

\author{Marius-F. Danca\and Michal Fe\v{c}kan \and Nikolay Kuznetsov}

\institute{Marius-F. Danca (corresponding author)\at
Dept. of Mathematics and Computer Science\\
Avram Iancu University, 400380 Cluj-Napoca, Romania and\\
Romanian Institute for Science and Technology
400487 Cluj-Napoca, Romania \\
\email{danca@rist.ro}
\and
Michal Fe\v{c}kan \at
Dept. of Mathematical Analysis and Numerical Mathematics, Comenius University in Bratislava, Slovak Republic and \\
Mathematical Institute, Slovak Academy of Sciences, Slovak Republic \\
\email{Michal.Feckan@fmph.uniba.sk}
\and
Nikolay Kuznetsov \at
Department of Applied Cybernetics, \\Saint-Petersburg State University, Russia and\\
Dept. of Mathematical Information Technology,
 University of Jyv\"{a}skyl\"{a}, Finland\\
\email{nkuznetsov239@gmail.com}}

\maketitle

\begin{abstract}In this paper the chaos control in the discrete logistic map of fractional order is obtained with an impulsive control algorithm. The underlying discrete initial value problem of fractional order is considered in terms of Caputo delta fractional difference. Every $\Delta$ steps, the state variable is instantly modified with the same impulse value, chosen from a bifurcation diagram versus impulse. It is shown that the solution of the impulsive control is bounded. The numerical results are verified via time series, histograms, and the 0-1 test. Several examples are considered.

\vspace{3mm}
\textbf{keyword} Caputo delta fractional difference; Impulsive chaos control; Discrete logistic map of fractional order; Lyapunov exponent of discrete maps of fractional order; 0-1 test
\end{abstract}

\section{Introduction}

The models which involve abruptly change of variables are called impulsive equations.

The concept of impulsive control has a long history. Many impulsive control methods were successfully developed under the framework of optimal control. In mechanical systems, impulsive phenomena had been studied for different scenarios such as mechanical systems with impacts.

The theory of impulsive differential equations studies systems evolution, when some process is interrupted by abrupt changes (impulses) of state \cite{mumu_01}. These systems are modeled by differential equations which describe the period of continuous variation of state and by conditions which describe the discontinuities of first kind of the solution or of its derivatives at the moments of impulses. Many real world problems can experience abrupt external forces which can change completely their dynamics. For instance, an example of a real world problem that can be represented by an impulsive differential equation is a medicine intake, where the user must take regular doses of the medicine, which causes abrupt changes in the amount of medicine in their body, to control the disease or making it disappear (examples of impulsive systems can be found in, e.g. \cite{mumu_02,mumu_03,mumu_04,mumu_05}).

Details concerning existence and uniqueness of solutions, dependence of solutions on initial values, variation of parameters, oscillation and stability can be found in \cite{mumu6,ben}.

There exist different kinds of impulses \cite{mumu_01}, for instance, systems with impulses applied at fixed times (presented first in \cite{mumu_08,mumu_09}) and systems with impulses applied at variable times \cite{mumu_010,mumu_011} (see also \cite{mumu_02,mumu_07}). Impulses applied at vary time are important due to their applicability in the real world problems. For example, the  billiard-type system can  be  modeled  by differential systems with impulses which act on the first derivatives of the solutions. Thus, the positions of the colliding balls do not change at the moments of impact (impulse), but their velocities gain finite increments (the velocity will change according to the position of the ball) \cite{mumu_01}.

In the last years, results arising from impulsive effects have been adapted easily to the discrete case.

Difference equations or discrete dynamical systems is a diverse field which impacts almost every branch of pure and applied mathematics.
Discrete systems with memory in population, economy price option and signal processing have been considered almost at the same time since the fractional differential models are used. To note that often real systems may encounter abrupt changes at certain time moments and therefore, cannot be considered continuously but discrete-time (see e.g. \cite{mumu_015}).

On the other side, the control of chaos, or  control of chaotic systems, is the boundary field between control theory and dynamical systems theory studying when and how it is possible  to control systems exhibiting irregular, chaotic behavior (see e.g. \cite{mumu_016}).

For discrete equations of integer order, the impulsive control algorithm utilized in this paper has the following form
\begin{equation}\label{asta}
x_{n+1}=\left\{
\begin{array}{l}
f(x_{n}),~~ \text{if}~~ \mod(n,\Delta)\neq0  \\
(1+\gamma)x_{n+1},~\text{if}~~ \mod(n,\Delta)=0
\end{array}
\right.,
\end{equation}
where $f\in C(\R,\R)$ is some discrete map which depend on a real bifurcation parameter, $\Delta \in N^*$ and the impulse $\gamma\in \mathbb{R}$ a relative small real number. One assumes that for some parameter ranges, the system evolves chaotic.

The algorithm perturbs $x$ every $\Delta$ steps with the quantity $(1+\gamma)$ and acts ``instantaneously'' in the sense it modifies $x_{n+1}$ while it is calculated, the system dynamics being subject to abrupt changes, $(1+\gamma)x_{n+1}$ (in nature, these can be shocks, harvesting, natural disasters, etc.). If, without impulses $\gamma$, for some parameter value the system behaves chaotically, an adequate design of the chaos control (i.e. adequate values of impulses $\gamma$ and time-moments $\Delta$) may force the system to become stable evolving along some regular trajectory.

The impulsive moment $\Delta$ is fixed in advance, while the impulse $\gamma$ is chosen from the bifurcation plot of $x$ versus $\gamma$ shows the periodic windows where $\gamma$ generates stable periodic orbits.

The impulsive control \eqref{asta} has important implications, for example, in ecology. Thus, the phenomenon of a population increasing in response to an increase in its per-capita mortality rate has to be taken into consideration to design strategies in fisheries and pest management. This paradoxical phenomenon is known as the hydra effect \cite{mumu_017}. The algorithm can also be applied successfully in chemical systems \cite{mumu_018}.

Impulsive equations modeled with continuous or discontinuous differential equations of integer order or fractional order \cite{sase,doi,trei,patru}, or by discrete equations \cite{unu} have been developed in impulsive problems in physics, orbital transfer of satellite, population dynamics,
dosage supply in pharmacokinetics, biotechnology, pharmacokinetics, ecosystems management, industrial robotics, synchronization in chaotic secure communication systems, and so forth (see \cite{ben} for a deep background on impulsive differential equations and inclusions and references, or \cite{mumu6}).

The discrete fractional calculus has been an increased interest due to its importance in real world problems. More generalized chaos does has been found in fractional discrete systems \cite{mumu_012,mumu2,mumu_013,mumu_014}.

The stability of impulsive fractional difference equations is studied in \cite{balax} and the first study of the fractional standard map with memory, derived from a differential equation, can be found in \cite{mumu} (see also \cite{mumu2,mumu3,mumu4,mumu5,golgol,golgol2}).

In this paper one considers the chaos control of the discrete logistic map of fractional order. It is proved that the impulsive solution of the controlled logistic map of fractional order is bounded. The numerical results are verified with the 0-1 test, time series, histograms and the Lyapunov exponent. Because of the memory history effect, the numerical determination of the Lyapunov exponent requires a special approach.

The paper is organized as follows: Section 2, deals with the discrete logistic map of fractional order, and the applicative Section 3 presents the numerical implementation of the chaos control algorithm \eqref{asta} in the case of the fractional logistic map of fractional order. The Appendix presents briefly the 0-1 test. The Conclusion section ends the manuscript.

\section{The discrete logistic map of fractional order}

Let $q\in(0,1)$, $\N_{1-q}=\{1-q,2-q,3-q,\cdots\}$, $0<q\leq1$ and $f\in C(\R,\R)$ a discrete map.

The difference equations of fractional order (FO) studied in this paper are modeled by the following initial value problem

\begin{equation}\label{e1}
\triangle_*^q x(k)=f(x(k-1+q)),\quad k\in \N_{1-q}, ~~x(0)=x_0,
\end{equation}
where $\triangle_*^q x(k)$ is the Caputo delta fractional difference \cite{sta1,FP} .

Hereafter, the FO equations are considered with initial condition $x(0)=x_0$.

The equivalent discrete integral form of \eqref{e1} is (see e.g. \cite{sta1,FP})
$$
x(n)=x(0)+\frac{1}{\Gamma(q)}\sum_{j=1-q}^{n-q}\frac{\Gamma(n-j)}{\Gamma(n-j-q)}f(x(j-1+q)),
$$
which with $j\leftrightarrow j+q$ becomes

\begin{equation}\label{e2}
x(n)=x(0)+\frac{1}{\Gamma(q)}\sum_{j=1}^{n}\frac{\Gamma(n-j+q)}{\Gamma(n-j+1)}f(x(j-1)),~~n=1,2,...
\end{equation}

Consider \eqref{e1} in the case of the discrete logistic map of FO \cite{bal2}
\begin{equation}\label{d0}
\triangle_*^q x(k)=f(x(k+q-1)):=\mu x(k+q-1)(1-x(k+q-1)), ~~k\in \mathbb{N}_{1-q}.
\end{equation}

Then, the underlying discrete integral \eqref{e2} becomes (see also \cite{bal2})

\begin{equation}\label{d1}
x(n)=x(0)+\frac{\mu}{\Gamma(q)}\sum_{j=1}^{n}\frac{\Gamma(n-j+q)}{\Gamma(n-j+1)}x(j-1)(1-x(j-1)).
\end{equation}

Because, due the discrete memory effect, the Jacobian matrix necessary for Lyapunov exponent (LE) cannot be obtain directly. Therefore, in \cite{bal1} is proposed the following natural way of linearization of \eqref{d1} along the orbit $x_n$

\begin{equation}\label{d2}
a(n)=a(0)+\frac{\mu}{\Gamma(q)}\sum_{j=1}^{n}\frac{\Gamma(n-j+q)}{\Gamma(n-j+1)}a(j-1)(1-2x(j-1)),~~a(0)=1,
\end{equation}
wherefrom, from \eqref{d2} via \eqref{d1}, the finite-time local LE, $\lambda$, is obtained as follows
\begin{equation*}
\lambda(x_0)\simeq\frac{1}{n}\ln|a(n-1)|.
\end{equation*}

\noindent Due to the discrete memory effect (the present status depends on the all previous information), one of the main impediments to implement \eqref{e2}, is the instability of
$$
R_q(n):=\sum_{j=1}^{n}\frac{\Gamma(n-j+q)}{\Gamma(n-j+1)},~~n=1,2,...
$$

Thus, one has the following
\begin{proposition}\label{propo}
\begin{equation}\label{e3}
\lim_{n\to\infty}\frac{1}{n^q}R_q(n)=\frac{1}{q}.
\end{equation}
\end{proposition}
\begin{proof}
By using Gautchi inequality \cite{G,K}
$$
\frac{1}{(x+1)^{1-q}}\le\frac{\Gamma(x+q)}{\Gamma(x+1)}\le \frac{1}{\left(x+\frac{q}{2}\right)^{1-q}},
$$
for any $x\ge0$, one derives
$$
\begin{gathered}
R_q(n)\ge\sum_{j=1}^{n}\frac{1}{(n-j+1)^{1-q}}=\sum_{j=1}^{n}\frac{1}{j^{1-q}}\\
\ge \int_1^{n+1}\frac{dx}{x^{1-q}}=\frac{1}{q}((n+1)^q-1)\ge\frac{1}{q}(n^q-1),
\end{gathered}
$$
and
$$
\begin{gathered}
R_q(n)\le\sum_{j=1}^{n}\frac{1}{\left(n-j+\frac{q}{2}\right)^{1-q}}
=\left(\frac{q}{2}\right)^{q-1}+\sum_{j=1}^{n-1}\frac{1}{\left(j+\frac{q}{2}\right)^{1-q}}\\
\le \left(\frac{q}{2}\right)^{q-1}+\int_0^{n-1}\frac{dx}{\left(x+\frac{q}{2}\right)^{1-q}}=\frac{1}{2}\left(\frac{q}{2}\right)^{q-1}
+\frac{1}{q}\left(n-1+\frac{q}{2}\right)^q\le \frac{1}{q}\left(n^q+1\right).
\end{gathered}
$$
Thus
\begin{equation*}
\left|R_q(n)-\frac{n^q}{q}\right|\le\frac{1}{q},\quad \forall n\ge1.
\end{equation*}
\end{proof}
Relation \eqref{e3} means that the terms of $(R_q(n))$ and $(n^q)$ grow similarly. Therefore, small errors in each steps may lead to large final errors.
In Fig. \ref{fig1} is presented the evolution of $R_q(n)$ for $n\le1500$ and $q=k0.1$, $k=1,2,...,10$\footnote{A possible way to extend the range of $n$ in the numerical determination of $R_q(n)$ in Matlab, is to use the following relation: $\mathtt{\Gamma(n-j+q)/\Gamma(n-j+1)=exp(gammaln(n-j+q)-gammaln(n-j+q))}$, where $\mathtt{gammaln}$ is the logarithm of the gamma function.}.

\begin{remark} \leavevmode
\begin{itemize}[leftmargin=0.2in,noitemsep,topsep=0pt]
\item [i)]
Denote by $\Phi(t,x_0)$, a solution of \eqref{e1}. Because, due to the memory history of solutions, $\Phi$ does not verify the relation $\Phi(t)\circ\Phi(s)=\Phi(t+s)$, one cannot consider \eqref{e1} as defining a dynamical system. On the other side, motivated by the numerical calculation of the solutions utilized in this paper, the definition of a dynamical system of integer order modeled by a numerical scheme \cite[Definition 2.1.2]{x2} is adopted. Therefore, because \eqref{d0} admits the solutions \eqref{d1}, one says the problem defines a dynamical system of FO.
\item[ii)] Formula \eqref{e2} is valid also for the following FO discrete equations
$$
\nabla_*^q x(k+1)=f(x(k)),\quad k=0,1,\cdots,\quad x(0)=x_0,
$$
where $\nabla_*^q x(k)$ is the Caputo nabla fractional difference (see \cite{sta1}).
\end{itemize}
\end{remark}

\section{Chaos control of the logistic map of FO}

The chaos control algorithm \eqref{asta} can be expressed as follows:
\begin{equation}\label{eee1}
x(n+1)=\begin{cases} x_0+\frac{1}{\Gamma(q)}\sum_{j=1}^{n+1}\frac{\Gamma(n-j+q)}{\Gamma(n-j+1)}f(x(j-1)) &\textrm{if } \mod(n,\Delta)\ne0,\\
                   (1+\gamma)(x_0+\frac{1}{\Gamma(q)}\sum_{j=1}^{n+1}\frac{\Gamma(n-j+q)}{\Gamma(n-j+1)}f(x(j-1))) &\textrm{if } \mod(n,\Delta)=0,
     \end{cases}
\end{equation}
where $f(x(n))$ is the logistic map.

The sequence $(x(n))$ is bounded, as proved by the following result

\begin{theorem}
For \eqref{eee1} with $\gamma+1>0$, it holds
$$
x(n)\le\max\Bigg\{x_0+\frac{((n+1)^q+1)\mu}{4q},(1+\gamma)\left(x_0+\frac{((n+1)^q+1)\mu}{4q}\right)\Bigg\}\quad\forall n\in\N.
$$

\end{theorem}
\begin{proof}
It is clear that $f(x)\le \frac{\mu}{4}$ for any $x\in\R$. Then we get
$$
\begin{gathered}
x(n+1)\le x_0+\frac{((n+1)^q+1)\mu}{4q},\quad \mod(n,\Delta)\ne0,\\
x(n+1)\le (1+\gamma)\left(x_0+\frac{((n+1)^q+1)\mu}{4q}\right),\quad \mod(n,\Delta)=0.
\end{gathered}
$$
\end{proof}

To implement the impulsive algorithm \eqref{eee1} and visualize the effect of impulses $\gamma$, one draws the bifurcation diagram of the impulsed system versus $\gamma\in[-0.15,0.15]$. The tools utilized to analyze the regularity obtained with the impulsive control are the LE, time series, the `0-1' test (Appendix) with $p$, $q$, the mean square displacement $M$ and the median $K$ value, and also histograms. First transients have been removed. For this system, the errors of $K$ are about $1e-3$ for the `0' value, and $1e-4$ for `1'.

Note that if $\gamma>0$, the system suffers a positive variation of internal energy, and the receives energy, because at the moment $\Delta$, when the impulse is applied, the new value of $x_{n+1}$, $(1+\gamma)x_{n+1}$, is bigger than the previous value \cite{doi}. The received energy is used to maintain the orbit on a regular motion. Similarly, if $\gamma<0$, the system is forced to dissipates energy to reach the necessary one along the regular orbit. However, these phenomena should be interpreted under the mentioned hystory memory effect.

The bifurcation diagram of the FO logistic map \eqref{d0} versus $\mu\in[1.8,2.8]$ with $q=0.9$ is presented in Fig. \ref{fig2} (a), while in Fig. \ref{fig2} (b) is presented the bifurcation plot versus $q\in[0.01,1]$ with $\mu=2.5$.
All bifurcation diagrams in this paper are overplotted with LE exponent (red plot) and $K$ median value (blue plot).

\begin{remark}\leavevmode\label{ics}
\begin{itemize}[leftmargin=0.2in,noitemsep,topsep=0pt]
\item[i)]
As known, because the memory effect, general differential equations of FO cannot have nonconstant periodic solutions (see e.g. \cite{inex}). Similarly this happens with discrete difference equations of FO\cite{mich}. Therefore, these orbits are called ``numerically stable periodic'' (NSP), in the sense that the trajectory, from numerical point of view, can be an extremely near periodic trajectory \cite{ooo}.
\item[ii)] Even the bifurcation scenario versus $\mu$ (Fig. \ref{fig2} (a)) looks similar to the case of integer-order logistic map, the Feigenbaum parameter sequence is different. Also, over $\mu>2.8$, depending on $q$, orbits may diverge (grey zones in Figs. \ref{fig5} and \ref{fig6}).
\item[iii)] As the detail in the bifurcation diagram of $x$ versus $q$ in Fig. \ref{fig2} (b) shows, the system does not present the typical periodic direct and reverse period doubling bifurcations (dotted red lines), as for the integer order counterpart (compare also with \cite[Figs. 4-7]{bal2} and \cite[Fig. 1]{pupu1}). Because at this points the LE is positive and not zero as in the case of the logistic map of integer order, but  $K$ is approximately $0$ like for regular motion, it is difficult to characterize this kind of dynamics. Similarly, but less visible, it happens in the bifurcation diagram versus $\mu$ (Fig. \ref{fig2} (a)).
\end{itemize}
\end{remark}

\noindent For the next numerical experiments $q=0.9$ and $\mu=2.5$ (Fig. \ref{fig2} (a)), when LE=0.1744 and $K=0.9538$ fact which indicates a chaotic behavior.

\noindent Next, several values of moments $\Delta$ are considered.
\begin{enumerate}[topsep=0pt]

\item $\Delta=1$. In this case impulses applied every step (Figs. \ref{fig3}). As can be seen, there are several NSP windows, corresponding to negative but also positive ranges of $\gamma$, from which each value $\gamma$ leads to NSP orbits.
Note the points corresponding to $\gamma_1$, $\gamma_2$ and $\gamma_3$, where LE$>0$ and $K\approx 0$ (see Remark \ref{ics} (iii)).
Two representative values are considered.
\begin{itemize}[leftmargin=*,labelindent=10pt,itemindent=20pt]
\item [-]
For $\gamma=-0.05$ (dotted line in Fig. \ref{fig3} (a)), with the impulsive control applied every step, one obtains a NSP orbit of period-$4$, as revealed by the time series (Fig. \ref{fig3} (b)) and histogram (Fig. \ref{fig3} (c)). The regularity of this NSP orbit is underlined by the circular plot of $p$ and $q$ (Fig. \ref{fig3} (d) and the bounded mean square displacement $M$ (Fig. \ref{fig3} (e). In this case LE=-0.0061 and $K=-0.0055$;
\item [-]
For $\gamma=0.059$ (dotted line in Fig. \ref{fig3} (a)), the impulsive control leads to a NSP of period-$5$ (see time series and historgram in Figs. \ref{fig4} (a) and (b) respectively). The regularity is underlined by the shape of the plot of $p$ and $q$ (Fig. \ref{fig4} (c)) and also by $M$ (Fig. \ref{fig4} (d)). LE=0.0471 and $K=-0.0076$.
\end{itemize}
\item The case $\Delta=3$ is presented in Figs. \ref{fig5}, where a NSP orbit of period-$12$ is obtained with $\gamma=-0.031$. The period can be remarked in the time series and histogram in Figs. \ref{fig5} (a) and (b), respectively, while the graphs of $p$ and $q$ (Fig. \ref{fig5} (c)) and the bounded mean square displacement $M$ show that the orbit is regular. LE=-0.0033 and $K=-0.0056$.
\item $\Delta_{max}$ for which chaos still can be controlled, for $\gamma\in[-0.15,0.15]$, is $\Delta_{max}=15$. Thus, for $\gamma=-0.12$ within a narrow periodic window $D$ (Fig. \ref{fig6} (a)), one obtains the NSP orbit of period-15 (Fig. \ref{fig6} (b)).

\item For $\Delta$ larger than $16$ and $\gamma\in[-0.15,0.15]$, chaos cannot be controlled. This is illustrated by the bifurcation diagram for $\Delta=17$, and the positive LE and values of $K$ close to 1 (Fig. \ref{fig6} (b)).
\end{enumerate}

\section*{Conclusion}
In this paper, the impulsive chaos control of the discrete logistic map of fractional order \eqref{eee1} has been investigated. The discrete logistic map of fractional order has been proposed by Wu and Baleanu in \cite{bal2} in terms of Caputo delta fractional difference.
The impulsive control,previously used in integer order continuous and discrete systems, is obtained by perturbing periodically (every $\Delta$ steps) the state variable with a constant impulse: $x_{n+1}\leftarrow (1+\gamma)x_{n+1}$, where $\gamma$ is a relatively small real number. If, for a chosen $\Delta$, the control algorithm is applied for a $\gamma$ value which generates in the bifurcation diagram versus $\gamma$ a chaotic behavior, regular motions can be obtained. Several numerical cases are considered.

It is proved that the impulsed orbits remain bounded.

To verify the obtained results, time series, histograms and the `0-1' test is utilized. Because of the discrete memory effect, the Lyapunov exponent is obtained by linearization of the discrete integral of the initial value problem of fractional order. Note that the numerical implementation of the discrete integral of the underlying initial value of FO requires numerical precaution. Otherwise only few terms of iterations can be calculated.

\textbf{Acknowledgement} N.K and M.F.D. are supported by the Russian Science Foundation 19-41-02002 and M.F. is supported by the Grants Slovak Research and Development Agency under the contract No. APVV-18-0308 and by the Slovak Grant Agency VEGA-SAV  No. 2/0153/16 and No. 1/0078/17.

Conflict of Interest: The authors declare that they have no conflict of interest.

\newpage
\section*{Appendix}

\subsection*{The `0-1' test}\label{a}
The '0-1'  test has its roots in \cite{uu} beening developed in \cite{got2} (see also \cite{got1} or \cite{uu2}). It is designed to distinguish chaotic behavior from regular behavior in deterministic systems.
The input being a time series, the test is easy to implement and does not need the system equations.
Consider a discrete or continuous-time dynamical system and a one-dimensional observable data set, constructed from a time series, $\phi(j)$, $j=1,2,...,N$, with $N$ some positive integer. The `0-1' test bases on a theorem, which states that a nonchaotic motion is bounded, while a chaotic dynamic behaves like a Brownian motion \cite{uu}.

\noindent 1) First, for $c\in[0,2\pi]$, one compute the translation variables $p$ and $q$ \cite{got2}

\[
p(n)=\sum_{j=1}^n\phi(j)\cos(jc),~~~ q(n)=\sum_{j=1}^n\phi(j)\sin(jc),
\]
for $n=1,2,...,N$. The choice of $c$ represents an important and sensible algorithm variable (see for example \cite {got1} where for $c$, the interval $[\pi/5,4\pi/5]$ is proposed).

\noindent 2) To determine the growths of $p$ and $q$, the mean-square displacement $M$ is determined:
\[
M(n)=\lim_{N\rightarrow \infty}\frac{1}{N}\sum_{j=1}^N[p(j+n)-p(j)]^2+[q(j+n)-q(j)]^2.
\]
where $n\ll N$ (in practice, $n=N/10$ represents a good choice).

\noindent 3) Next, the asymptotic growth rate $K$ is defined as
\[
K=\lim_{n\rightarrow \infty}\log M(n)/\log n.
\]
If the underlying dynamics is regular (i.e. periodic or quasiperiodic) then $K \approx 0$; if the underlying dynamics is chaotic then $K \approx 1$.

In Fig. \ref{fig7} the case of the logistic map of integer-order is presented. In Figs. \ref{fig7} (a) are presented the plots of $q$ and $p$ while in Figs. \ref{fig7} (b) the mean-square displacement $M$ as a function of $n$. In Figs. (i) the regular orbit of the logistic map $x_{n+1}=\mu x_n(1-x_n)$ for $\mu=3.55$ while Figs. (ii) present the chaotic orbit of the logistic map for $\mu=4$.

\begin{figure}
\begin{center}
\includegraphics[scale=1]{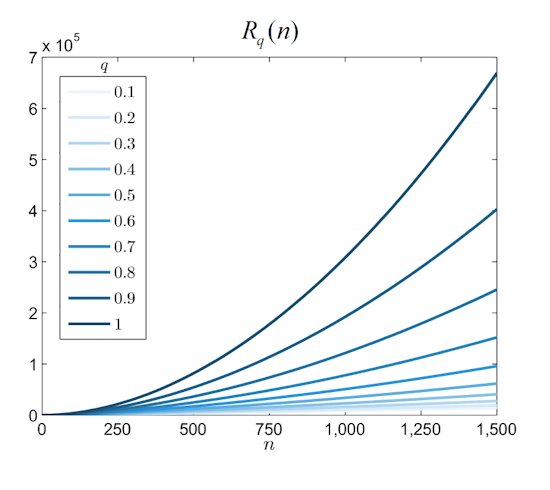}
\caption{Variation of $R_q(n)$, for $n\in[0,1500]$, and $q\in\{0.1,0.2,...,1\}$. }
\label{fig1}
\end{center}
\end{figure}

\begin{figure}
\begin{center}
\includegraphics[scale=1]{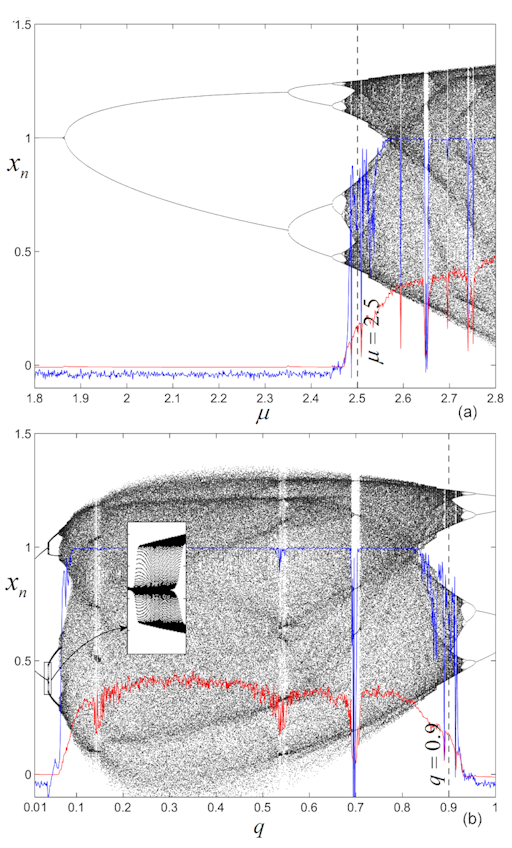}
\caption{Bifurcation diagrams of the FO logistic map \eqref{d1}. (a) Bifurcation diagram versus $\mu\in[1.8,2.8]$ with $q=0.9$; (b) Bifurcation diagram versus the FO $q\in[0.01,1]$ with $\mu=2.5$.}
\label{fig2}
\end{center}
\end{figure}

\begin{figure}
\begin{center}
\includegraphics[scale=1]{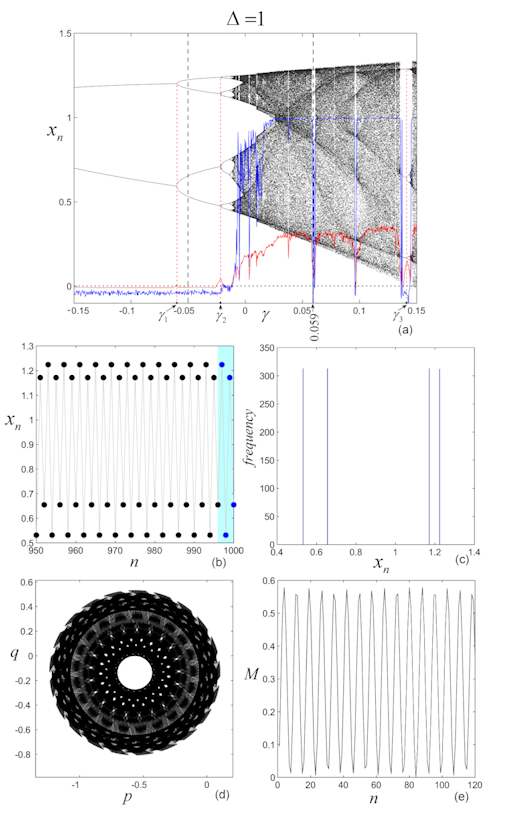}
\caption{NSP period-$4$ orbit obtained with the algorithm \eqref{eee1} applied every step ($\Delta=1$), for $\gamma=-0.05$; (a) Bifurcation diagram versus $\gamma\in[-0.15,0.15]$. Two values of $\gamma$ for which the system behaves regularly have been chosen (dotted lines): $\gamma=-0.05$ and $\gamma=0.059$; (b) Time series for $\gamma=-0.05$ indicates a NSP orbit of period-$4$; (c) Histogram for $\gamma=-0.05$ with $4$ bars which indicates a NSP orbit of period-$4$ ; (d) $q$ and $p$ plot for $\gamma=-0.05$; (e) mean square displacement $M$. }
\label{fig3}
\end{center}
\end{figure}

\begin{figure}
\begin{center}
\includegraphics[scale=1]{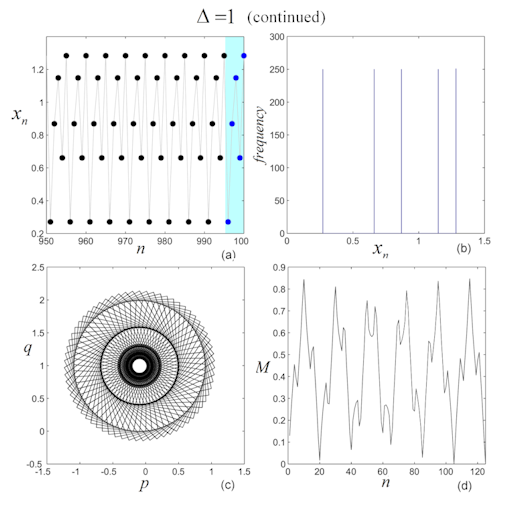}
\caption{A period-$5$ NSP orbit, obtained with the algorithm \eqref{eee1}, with $\Delta=1$, but for $\gamma$ chosen within other stable window: $\gamma=0.059$ (see \ref{fig3} (a)). (a) Time series indicates a NSP orbit of period-$5$; (b) Histogram with $5$ bars which indicates a NSP orbit of period-$5$; (c) $q$ and $p$ plot; (d) mean square displacement $M$.}
\label{fig4}
\end{center}
\end{figure}

\begin{figure}
\begin{center}
\includegraphics[scale=1]{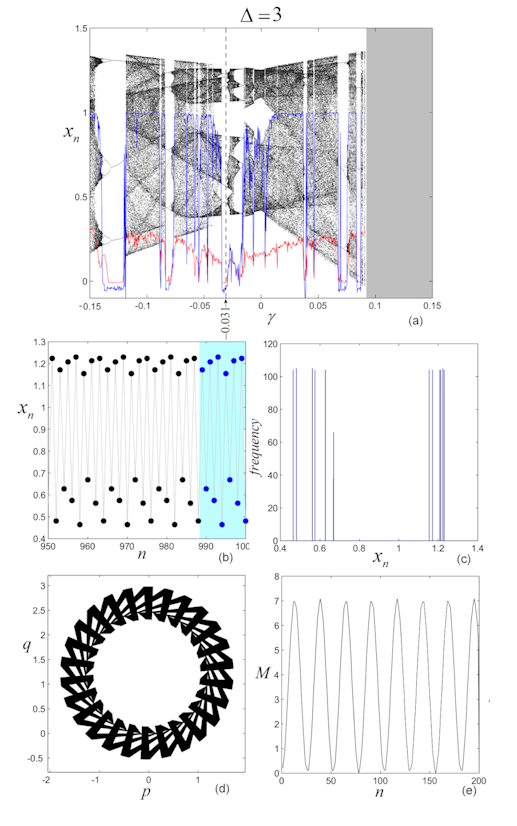}
\caption{NSP orbit of period-$12$ obtained with the impulsive chaos control \eqref{eee1}, applied every three step ($\Delta=3$) and $\gamma=-0.031$ (dotted line in the bifurcation diagram) (a) Bifurcation diagram versus $\gamma\in[-0.15,0.15]$; (b) Time series indicating a NSP orbit of period-$12$; (c) Histogram with $12$ bars which indicates a NSP orbit of period-$12$; (d) $q$ and $p$ plot; (e) mean square displacement $M$.}

\label{fig5}
\end{center}
\end{figure}

\begin{figure}
\begin{center}
\includegraphics[scale=1]{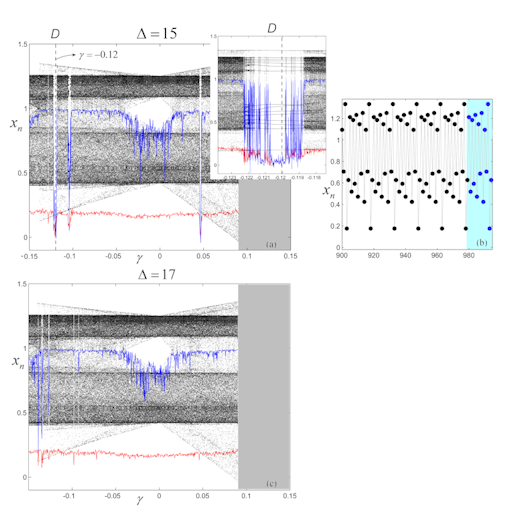}
\caption{(a) NSP orbit of period-$15$ obtained with the algorithm \eqref{eee1} applied $\Delta=1$ steps and $\gamma=-0.12$ (dotted line); $D$ represents a zoomed zone of the  chosen $\gamma$; (b) Time series indicating a NSP orbit of period-$15$; (c) Bifurcation diagram for $\Delta=17$. LE has positive values, and $K$ values close to 1.}
\label{fig6}
\end{center}
\end{figure}

\begin{figure}
\begin{center}
\includegraphics[scale=1]{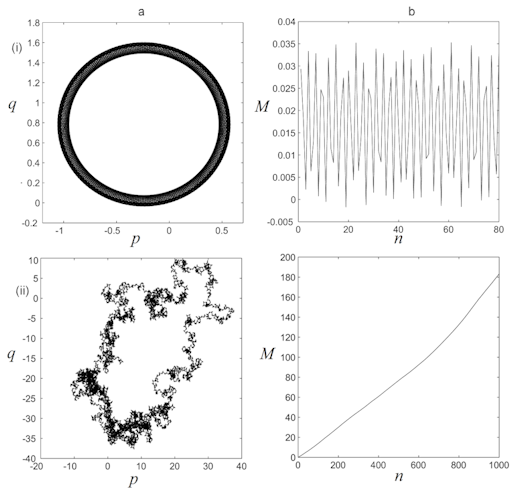}
\caption{The `0-1' test applied to the logistic map of integer order $x_{n+1}=\mu x_n(1-x_n)$; (i) Periodic motion for $\mu=3.5$; (ii) Chaotic motion for $\mu=4$; a: Plot of $q$ and $p$; b: The mean square displacement $M$.}
\label{fig7}
\end{center}
\end{figure}

The authors declare that they have no conflict of interest.

\newpage{\pagestyle{empty}\cleardoublepage}


\end{document}